\documentclass[11pt,draftcls,onecolumn,english]{IEEEtran}
\usepackage[T1]{fontenc}
\usepackage[latin9]{inputenc}
\usepackage{amsthm}
\usepackage{amsmath}
\usepackage{graphicx}
\usepackage{setspace}
\usepackage{amssymb}
\usepackage{psfrag}
\usepackage{subfigure}
\usepackage{multicol}
\usepackage{geometry}
\usepackage{url}

\usepackage{amsthm}
\usepackage{cite}
\usepackage{amsthm}
\usepackage{epsfig}
\IEEEoverridecommandlockouts
\newtheorem{thm}{Theorem}
\newtheorem{prop}[thm]{Proposition}
\newtheorem{cor}[thm]{Corollary}
\newtheorem{fact}[thm]{Fact}
\theoremstyle{definition}
\newtheorem{defn}[thm]{Definition}
\theoremstyle{plain}

\makeatletter

\makeatletter

\makeatletter

\usepackage{amsthm}

\usepackage{epsfig}

\makeatother

\makeatother
\usepackage{babel}

\begin{document}

\title{Convergence Speed of the Consensus Algorithm with Interference and Sparse Long-Range Connectivity}

\author{\begin{table}[hbt]
		\begin{center}
			\begin{tabular}{lll}
				\small S.~Vanka$^{*}$ & \small M.~Haenggi & \small V.~Gupta \\
				\small Department of Electrical Engineering & \small Department of Electrical Engineering & \small Department of Electrical Engineering \\
				\small 218 Cushing Hall & \small 274 Fitzpatrick Hall & \small 270 Fitzpatrick Hall \\
				\small University of Notre Dame & \small University of Notre Dame & \small University of Notre Dame \\
				\small Notre Dame, IN 46556, USA & \small Notre Dame, IN 46556, USA & \small Notre Dame, IN 46556, USA \\ 
				\small E-mail:svanka@nd.edu & \small E-mail:mhaenggi@nd.edu & \small E-mail:vgupta2@nd.edu \\
				\small Phone:(574)-631-8245 & \small (574)-631-6103 & \small (574)-631-2294 
			\end{tabular}
		\end{center}
        \end{table}
\thanks{{*}Corresponding author. The work of the first two authors was partially supported by NSF (grants CNS 04-47869 and CCF 728763). The work of the third author was supported partially by the NSF award 0846631.}%
}
\maketitle
\IEEEpeerreviewmaketitle
\vspace{-3cm}
\begin{abstract}
We analyze the effect of interference on the convergence rate of average consensus algorithms, which iteratively compute the measurement average by message passing among nodes. It is usually assumed that these algorithms converge faster with a greater exchange of information (i.e., by increased network connectivity) in every iteration. However, when interference is taken into account, it is no longer clear if the rate of convergence increases with network connectivity. We study this problem for randomly-placed consensus-seeking nodes connected through an interference-limited network. We investigate the following questions: (a) How does the rate of convergence vary with increasing communication range of each node? and (b) How does this result change when each node is allowed to communicate with a few selected far-off nodes? When nodes schedule their transmissions to avoid interference, we show that the convergence speed scales with $r^{2-d}$, where $r$ is the communication range and $d$ is the number of dimensions. This scaling is the result of two competing effects when increasing $r$: Increased schedule length for interference-free transmission vs.~the speed gain due to improved connectivity. Hence, although one-dimensional networks can converge faster from a greater communication range despite increased interference, the two effects exactly offset one another in two-dimensions. In higher dimensions, increasing the communication range can actually degrade the rate of convergence. Our results thus underline the importance of factoring in the effect of interference in the design of distributed estimation algorithms.
\end{abstract}

\section{Introduction}
\subsection{Motivation}
The advent of wireless sensor and ad hoc networks has motivated the need for distributed information processing algorithms, which allow each node to operate only on local information. A well-studied algorithm that allows distributed averaging is the average consensus algorithm, wherein the global average of a set of initial sensor observations can be computed based on purely local computations at each sensor. Starting from a set of initial measurements, the average consensus algorithm allows a set of nodes to communicate by a (possibly time-varying) topology to iteratively compute the global average of the initial measurements, see e.g.,~\cite{XiaoBoyd03, Boyd06, Moallemi06, Blondel05, XiaoIPSN05, KarMoura08, Dimakis06, LiDaiZhangLADAU_IT10, AysalDisturb10} and references therein. The connectivity properties of the topologies that ensure convergence have been well-studied (e.g.,~\cite{Olfati-Saber04, RenBeard05}). Of late, the focus has shifted to studying convergence in the face of communication constraints, like quantization~\cite{Nedic09, Aysal08, KashyapQuantization07}, packet drops~\cite{Fagnani09} and noise~\cite{RajWainwright08}. A closely associated algorithm is the gossip algorithm \cite{AysalScaglioneBroadcastTSP09, Boyd06, AysalWireless2010 }. In particular, the recent work~\cite{AysalWireless2010} proposes and studies a probabilistic version of the broadcast gossip algorithm~\cite{AysalScaglioneBroadcastTSP09}. The idea is to exploit channel fluctuations to enable opportunistic longer-range message-passing. Since only one node is allowed to transmit at any given time, the question of interference does not arise.

In this paper, unlike prior work, we study the effect of interference, which becomes important in the formation of more general message-passing topologies. We explicitly model the effect of interference on the rate of topology formation---and hence convergence---of the average consensus algorithm. This important effect---which crucially depends on network geometry---has been largely ignored. In wireless networks, depending on the physical proximity of $a$ to $d$ and $c$ to $b$, the transmission from $a$ to $b$ and $c$ to $d$ may interfere with one another; hence two time slots may be needed to establish edges $\overrightarrow{(a,b)}$ and $\overrightarrow{(c,d)}$. The network thus has two time-scales of interest: that of establishing individual communications among the desired set of nodes and that of the iterations of the distributed algorithms, which occur only when all the desired nodes have successfully communicated. One may thus, view the underlying communication network as constructing the desired message passing graphs from several feasible sub-graphs, each of which satisfies half-duplex, fading and interference constraints. The union of all these sub-graphs is the desired message passing graph. 

To illustrate this, consider the formation of a simple linear 6-node network shown in Fig.~\ref{fig:SixNodeExample}. Suppose the estimation algorithm requires nearest-neighbor communication (shown as bidirectional edges). However, due to interference constraints, only every third node can transmit. In this case, we see that forming the the desired topology requires at least three time-slots, as shown. In other words, for these interference constraints, this topology's \emph{fastest rate of formation} is three time slots. Clearly, a topology's intrinsic benefit \emph{and} the fastest rate of its formation determine its true utility. 

\begin{figure}[hbt]
\centering
	\includegraphics[scale=0.35]{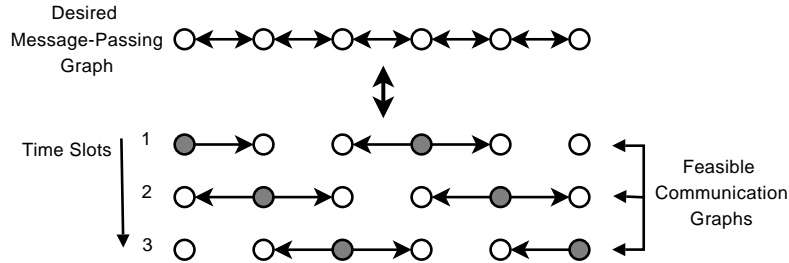}
	\caption{\label{fig:SixNodeExample}An example illustrating the constraints introduced by interference. If the nodes are physically placed as shown, interference limits the number of nodes that can communicate concurrently. Assuming a spatial re-use factor of two, the message-passing graph can be formed as a union of three feasible sub-graphs, each of them satisfying interfefence constraints. We consider this TDMA schedule \emph{feasible}.}
\end{figure}

The performance of the underlying (real-time) estimation algorithm is therefore \emph{coupled} with algorithms for channel access and routing. In our previous work~\cite{VGH08}, we studied the coupling with channel access for the average consensus algorithm for a certain class of deterministic network topologies. Using a simple protocol model~\cite{GK2000} for reception, we were able to show that the effect of increasing network connectivity depends crucially on its dimension. In our recent work~\cite{Vanka09globecom} we exploited the well-known parallels between the convergence of the average consensus algorithm and Markov chain mixing (e.g.,~\cite{Boyd06} and the references therein) to study consensus on disk graphs~\cite{Gilbert61} using the more refined physical model. We examined the scaling behavior of the fastest rate of topology formation with interference, captured by the shortest feasible TDMA schedules that construct the graph.

We note here that implementing inter-node communication in a network will require some additional overhead. For example, one possible protocol that establishes point-to-point communication can have nodes tag their packets with their uniquely assigned address. A receiver reads this address and decodes a packet only if the address is that of one of its intended transmitters. In this work, we neglect this additional overhead. However, we show that even when this overhead is neglected, increased interference alone is enough to significantly lower the rate of topology formation.

\subsection{Main Contributions}

In this paper, we study networks with short-range and networks with both short-range and limited long-range communication. Although remarkable improvements in convergence rate have been reported~\cite{SalehiSmallWorld07, AldosariMouraSmallWorld05, Olfati-SaberUltrafast} for consensus on graphs with a few long-range edges (as in small-world graphs~\cite{NewmanWatts00}),  it is not clear if these benefits will carry over to a wireless setting, where long-range links come at a cost of increased interference. Motivated by this fact, we study the average consensus problem in graphs formed by overlaying long-range edges onto an existing ``short-range'' disk graph. We derive the scaling law for the spectral gap as well as that of the fastest rate of topology formation in the presence of interference. To the best of our knowledge, this is the first such attempt.

We find that the spectral gap scales quadratically in the communication range $r$, independently of the network dimension $d$, but the length of the shortest TDMA schedule that constructs such graphs scales as $r^d$. Thus when interference is factored in, the benefit of a greater communication range depends crucially on the network dimension:
\begin{itemize}
\item For one-dimensional networks ($d=1$), topologies with increased communication range can converge faster despite greater interference. 
\item For two-dimensional networks, the rate of convergence scales \emph{independently} of the communication range. 
\item For three- (and higher-) dimensional networks, increasing the communication range can actually slow down convergence. 
\end{itemize}
Furthermore, these results hold whether each node only communicates with all other nodes within its communication range, or, additionally, with a small number of far-away nodes. Thus our results significantly change many optimistic results obtained by analyzing the consensus problem in an abstract graph-theoretic setting.

The remainder of this paper is organized as follows. In Section~\ref{sec:Definitions}, we provide some standard definitions and results used in this paper. In Section~\ref{sec:Problem-Formulation}, we specify our system model and formulate the problem using the terminology developed in Section~\ref{sec:Definitions}. In Section~\ref{sec:Cvg-Flat-Networks}, we discuss convergence results for the disk graph model. In Section~\ref{sec:Cvg-Tiered-Networks}. we study the effect of selective long-range communication and provide the relevant scaling results. Section~\ref{sec:Conclusions} concludes the paper.

\section{\label{sec:Definitions}Definitions and Notation}

To make this paper self-contained, we formally state the following standard definitions and facts about Markov chains and introduce some notation and 
other relevant terminology. 

\subsubsection{Basic Definitions from Markov Chain Theory}

Consider a connected undirected graph $G$, with $n$
vertices $V=\{1,2,\ldots,n\}$ and a set of edges $E$. We assume
$G$ also contains all self-loops, i.e., $i\in V\implies(i,i)\in E$.
Let $d_{i}$ denote the degree of vertex $i$. For more information, see \cite{LevinMixingTimesBook08}.

\begin{defn}(Random walk on a graph) A random walk $\mathcal{X}(G)=(X_{k})_{k\in\mathbb{Z}}$
on $V$ is characterized by the $n\times n$ transition probability
matrix $\mathbf{P}(G)=[p_{ij}]$, with $p_{ij}\triangleq\mathbb{P}(X_{k+1}=i\mid X_{k}=j)$,
and $p_{ij}>0$ only if $(i,j)\in E$, with $\sum_{j}p_{ij}=1$ $\forall i\in V$.\end{defn}

Observe that $P$ is stochastic.

\begin{defn}(Symmetric random walk) A random walk is symmetric if
$p_{ij}=p_{ji}$.\end{defn}

For a symmetric random walk $\mathbf{P}$ is doubly stochastic. 

\begin{fact}\label{fact:SymRandWalk}A random walk on $G$ is a Markov
chain with state space $V$. Given an initial distribution  $\boldsymbol{\pi}(0)$ over $V$, the distribution $\boldsymbol{\pi}(k+1)$ after $k+1$ steps
satisfies $\boldsymbol{\pi}(k+1)=\mathbf{P}\boldsymbol{\pi}(k)$ for $k=0,1,\ldots$\end{fact}

\begin{defn}(Stationary distribution of a Markov chain) A stationary distribution $\boldsymbol{\pi}^{*}$ satisfies $\boldsymbol{\pi}^{*}=\mathbf{P}\boldsymbol{\pi}^{*}$, i.e., remains invariant with time.\end{defn}

\begin{defn}(Reversible Markov chain) A Markov chain $\mathcal{X}=(X_{k})_{k\in\mathbb{Z}}$
is said to be \emph{reversible} if for all states $i$, $\pi_{j}^{*}p_{ij}=\pi_{i}^{*}p_{ji}$.
\end{defn}

\begin{fact}\label{fact:IrredAperiodUnique}An irreducible and aperiodic
Markov chain has a unique stationary distribution.\inputencoding{latin1}{
}\inputencoding{latin9}\end{fact}

\begin{defn}\label{def:NatRandWalk}(Natural random walk) A natural
random walk on $G$ is a random walk with \[
p_{ij}=\begin{cases}
1/2d_{i}, & (i,j)\in E,i\neq j\\
1/2 & i=j.\end{cases}\]
 \end{defn}

\begin{fact}\label{fact:NaturalRegularUniform}The natural random
walk is reversible, irreducible and aperiodic with a unique stationary
distribution $\pi_{i}^{*}=\frac{d_{i}}{\sum_{i}d_{i}}.$ When $G$
is regular, a natural random walk is also symmetric and has a uniform
stationary distribution.\end{fact}

\begin{defn}\label{def:MixingTime}(Mixing time of a random walk)
For a random walk $\mathcal{X}$ with a unique stationary distribution
$\pi^{*}$, consider the Total Variational (TV) distance\footnote{The TV distance between two distributions $\mu$ and $\nu$ over a countable set $\mathcal{S}$
is defined as $\|\mu-\nu\|_{TV}\triangleq\frac{1}{2}\sum_{i\in\mathcal{S}}|\mu_{i}-\nu_{i}|$
(essentially the $\ell_{1}$ norm).} (cf. \cite[Chap. 4]{LevinMixingTimesBook08}) $d_{\mathrm{TV},i}(t;\pi_{0})\triangleq\frac{1}{2}\sum_{i}|\mathbb{P}(X_{t}=i,\pi_{0})-\pi_{i}^{*}|$
for an initial distribution $\pi_{0}$. Then the mixing time of $\mathcal{X}$ 
is defined as \[
T_{\mathrm{mix}}(\epsilon;P)\triangleq\sup_{\pi(0)}\inf\{t:d_{\mathrm{TV}}(t;P,\pi(0))\leq\epsilon\}.\]
 \end{defn}

\subsubsection{Asymptotic Notation}

We use the following asymptotic notation. For two functions $f$ and
$g$ of a variable $n$, as $n\rightarrow\infty$, we write 
\begin{itemize}
\item $g=\mathcal{O}(f)$ if the ratio $g/f$ is asymptotically finite. Further, $g=o(f)$ if this limit is zero. 
\item $g=\Omega(f)$ if $f=\mathcal{O}(g)$. Further, $g=\omega(f)\iff f=o(g)$. 
\item $g=\Theta(f)$ if $g=\mathcal{O}(f)$ and $g=\Omega(f)$. 
\end{itemize}
When $f$ and $g$ are random, these relations are defined to hold with probability one. 


\subsubsection{\label{sub:Graph-Process}Graph Sequences and the Asymptotic
Regime}

Consider a sequence of (possibly random) undirected graphs $(G_{n})$,
whose $n^{\mathrm{th}}$ member $G_{n}$ has $n$ vertices $V_{n}=\{1,2,\ldots,n\}$
and a set of edges $E_{n}$. We assume each graph contains all self-loops.
Denote the maximum and minimum node degrees of $G_{n}$ by $d_{\max}(G_{n})$
(shortened to $d_{\max}$) and $d_{\min}$ (shortened to $d_{\min}$)
respectively. We provide some standard definitions below.

\begin{defn} 
(Asymptotically regular graph) $G_{n}$ is asymptotically regular
if $d_{\mathrm{max}}(G_{n})-d_{\mathrm{min}}(G_{n})=o(1)$.
\end{defn}
\begin{defn} 
(Asymptotically almost sure validity) A property $\mathcal{P}$ is true asymptotically almost surely (a.a.s.) for
a sequence of random objects $(X_{n})$, if $\lim_{n\rightarrow\infty}\mathbb{P}\left(X_{n}\,\text{has property }\mathcal{P}\right)=1$.
\end{defn}
We obtain scaling results for the convergence of the average consensus algorithm in large networks by mapping the problem to the scaling
of mixing times of natural random walks on a sequence of graphs that are connected and regular asymptotically almost surely.
\section{\label{sec:Problem-Formulation}Problem Formulation}

\subsection{Average Consensus and Random Walks\label{sub:Consensus-RandomWalks}}
Consider a set of sensor nodes $V_n=\{1,2,\ldots,n\}$. Associate with the $i^{\mathrm{th}}$ sensor an initial observation $z_i(0)$. Given a realization of a random  \emph{message-passing} graph $G_{n}$ with vertices $V_n$ and edges $E_n$, suppose that all the vertices $i\in V_n$ synchronously update their observations as \begin{equation}
z_{i}(k+1)=\frac{1}{2}z_{i}(k)+\frac{1}{2d_{i}}\sum_{j\in N_{i}(G_{n})}(z_{j}(k)-z_{i}(k)),\qquad k = 0,1,\ldots\label{eq:ScalarStateUpdate}\end{equation}
Here $N_{i}(G_{n})$ denotes the neighborhood of vertex $i$ in $G_n$. By stacking the individual
observations $z_{i}$ to form the observation vector $\mathbf{z}$, the $(k+1)^{\mathrm{th}}$ update starting from an initial observation vector $\mathbf{z}(0)$ can be written as
\begin{equation}
\mathbf{z}(k+1)=\mathbf{W}\!_n \mathbf{z}(k).\label{eq:ConsVectUpdate}\end{equation}
 where we have defined the \emph{update matrix} $\mathbf{W}\!_n\triangleq(\mathbf{I}_n-\mathbf{\Delta}_n\mathbf{L}_n)/2$, where $\mathbf{I}_n$ denotes the $n\times n$ identity matrix, $\mathbf{\Delta}_n\triangleq\mathrm{diag}[d_{i}^{-1}]$ and $\mathbf{L}_n$ is the graph Laplacian. Notice that $\mathbf{W}\!_n$ depends on the realization of the random graph $G_n$, which remains the same for all iterations. We will analyze the speed of convergence for specific families of random graphs in the scaling limit $n\rightarrow\infty$, by deriving properties of interest that hold a.a.s. for all realizations of $G_n$. 

Without loss of generality, let $z_{i}(0)>0$, and define $z_{i}^{\prime}(0)\triangleq z_{i}(0)/\sum_{i}z_{i}(0)$
as the normalized initial observation vector. In the light of Fact \ref{fact:SymRandWalk}
and Definition \ref{def:NatRandWalk}, the iteration $\mathbf{z}^{\prime}(k+1)=\mathbf{W}\!_n\mathbf{z}^{\prime}(k)$
can now be interpreted as time-evolution of the node occupancy distribution
of a natural random walk over $G_{n}$ with a transition probability
matrix $\mathbf{W}\!_n$ \cite{Boyd06,SalehiSmallWorld07}. 

If $G_{n}$ is also connected, this equivalence with a natural random walk ensures (from
Fact \ref{fact:IrredAperiodUnique}) that the value of each vertex
asymptotically reaches $\frac{1}{n}\sum_{i}z_{i}(0)=\mathbf{1}^{\mathrm{T}}\frac{\mathbf{z}(0)}{n}$
(a more general result for a time-varying case was studied in \cite{Olfati-Saber04}).
Interpreting each vertex as a sensor and the initial values $(z_{i}(0))_{i\in V_n}$ as sensor measurements, this algorithm allows each sensor to
iteratively compute the average $\frac{1}{n}\sum_{i}z_{i}(0)$ of the initial measurement set 
by exchanging messages as described in (\ref{eq:ScalarStateUpdate}).
We will sometimes also refer to $G_{n}$ as the \emph{message-passing network}.

The rate of convergence of (\ref{eq:ConsVectUpdate}) to its steady
state value can be understood in terms of the mixing time of the natural
random walk described by $\mathbf{W}\!_n$. Indeed, by expressing $z_{i}^{\prime}$
in terms of $z_{i}$, we can write from Definition
\ref{def:MixingTime}: \begin{equation}
T_{\mathrm{mix}}(\epsilon;\mathbf{W}\!_n)=\sup_{\mathbf{z(0)}}\inf\{k:\|\mathbf{z}(k)-n^{-1}\mathbf{1}{z}_0\|_{\mathrm{TV}}\leq\epsilon z_{0}\}\label{eq:EpsCvgTime}\end{equation}
 where $z_{0}\triangleq\sum_{i}z_{i}(0)$. 

When $G_{n}$ is a.a.s. connected and regular, we know from Fact \ref{fact:NaturalRegularUniform} that the stationary distribution of the random walk is uniform a.a.s., thereby implying convergence to average consensus a.a.s. 

In this paper, we analyze random graphs based on the disk graph~\cite{Gilbert61}, which are parameterized by the disk radius (see Section \ref{sub:Network-Models}). For this family of graphs, it is well-known that the graphs are a.a.s. connected if and only if the radius remains large enough with $n$ (i.e., in the  ``supercritical'' regime~\cite{Penrose03}, see, e.g.,~\cite{GuptaKumarCriticalPower98} for a proof). In this regime, the asymptotic regularity property was formally shown to hold a.a.s. in ~\cite[Lemma 10]{Boyd06}. In fact, in \cite{Boyd06} these two properties were used to establish scaling laws for the mixing time of both the natural and the fastest mixing reversible random walks on these graphs to the uniform distribution. 

It is well-known that the mixing time of a random walk can be characterized by the second-largest eigenvalue
of $\mathbf{W}\!_n$. Denoting the eigenvalues of $\mathbf{W}\!_n$ by $\mu_{1}=1>\mu_{2}>\cdots>\mu_{n}>0$, the asymptotic convergence of the iteration (\ref{eq:ConsVectUpdate}) is determined by $\mu_{2}$. The result below formally establishes this dependence: 

\begin{thm}\label{thm:MixingTimeBounds}\cite{Sinclair92}. The $\epsilon-$mixing
time of a random walk with a doubly stochastic positive definite transition
matrix $\mathbf{W}\!_n$ on a connected graph $G_{n}$ is bounded as\[
\frac{\mu_{2}\log(2\epsilon)^{-1}}{2(1-\mu_{2})}\leq T_{\mathrm{mix}}(\epsilon;\mathbf{W}\!_n)\leq\frac{\log n-\log\epsilon}{1-\mu_{2}},\]

where $1-\mu_{2}$ is called the \emph{spectral gap} of $G_{n}$.
\end{thm}
\emph{Remark:} Observe that the spectral gap controls the mixing time. In the scaling limit $n\rightarrow\infty$, the scaling of $\epsilon$ also becomes important. The logarithmic dependence on $\epsilon^{-1}$ suggests three meaningful possibilities: 
\begin{enumerate}
	\item Polynomial scaling:~$\epsilon=1/n^{\delta}$ for some fixed $\delta>0$. 
	\item Exponential scaling:~$\epsilon=\exp(-\delta^{\prime}n)$ for some fixed $\delta^\prime>0$.
	\item Constant error:~$\epsilon\ll1$ is constant.
\end{enumerate}
For polynomial and exponential error scaling, it is clear that the bounds in Theorem \ref{thm:MixingTimeBounds} are of the same order, and are $\Theta((1-\mu_2)^{-1}\log n)$ and $\Theta((1-\mu_2)^{-1}n)$ respectively. For constant error, the upper bound scales $\log n$ times faster than the lower bound, i.e., $T_{\mathrm{mix}}=\Omega((1-\mu_2)^{-1})$ and $T=\mathcal{O}((1-\mu_2)^{-1}\log n)$. 

In the sequel we assume polynomial scaling, as was done in~\cite{Boyd06}. It will become clear in the later sections that the scaling laws for exponential scaling follow from a substitution $\log n\mapsto n$.  

\emph{Spectral Gap and Cheeger's Inequality\label{sec:SpecGapIntro}}:

Intuition suggests that the mixing time of a Markov chain depends
on how {}``easy'' it is to move out of any specified region in the
state space. This property can be formalized with the notion
of \emph{conductance}. The conductance of a reversible Markov chain
on a state space $\Omega=V$ on a graph $G_{n}$ with an equilibrium
distribution $\pi^{*}$ is defined as follows~\cite{SinclairJerrum89}:

\begin{equation}
h=\min_{S\subset\Omega,\pi^{*}(S)\leq1/2}\frac{Q(S,\bar{S})}{\pi^{*}(S)},\label{eq:ConductanceDef}\end{equation}
where $\pi^{*}(S)\triangleq\sum_{i\in S}\pi^{*}(i)$ and $\bar{S}=\Omega\backslash S$,
and $Q(S,\bar{S})\triangleq\sum_{i\in s,j\in\bar{S}}\pi^{*}(i)\mathbb{P}(X_{n+1}=j|X_{n}=i)$.
Viewed in graph-theoretic terms, the numerator (\ref{eq:ConductanceDef})
measures the effective weighted flow across the cut $(S,\bar{S})$,
while the denominator measures the weighted capacity of $S$. Intuitively,
we would expect a larger conductance to correspond to a smaller mixing
time, or equivalently from Theorem \ref{thm:MixingTimeBounds}, a
larger $1-\mu_{2}$ of the underlying graph $G_{n}$. This is indeed
the case, as Cheeger's Inequality shows:

\begin{thm}\label{thm:CheegersInequality}{~\cite{Sinclair92}}. The spectral
gap of a reversible Markov chain satisfies\[
\frac{h^{2}}{2}\leq1-\mu_{2}\leq2h,\]

where $h$ is the conductance of the Markov chain. \end{thm}

Once we know how $h$ scales with $n$ for a (random) sequence of graphs $(G_{n})$, we can use Theorem \ref{thm:CheegersInequality}
to find the scaling law for their spectral gap. This, in turn, permits the use Theorem \ref{thm:MixingTimeBounds} in deriving scaling laws
for the mixing time for iterations of the form (\ref{eq:ConsVectUpdate}) on these sequences of graphs. In the following, motivated by the need
to capture the distance-dependence and randomness in the connectivity of the nodes, we present random geometric graph models for $G_{n}$.

\subsection{\label{sub:Network-Models}Network Models}

Each point $i\in\{1,2,\ldots,n\}$ is placed uniformly randomly
in a $d-$dimensional torus $\mathcal{T}_{d}$ on $[0,1]^{d}$, i.e.,
the vertices form a binomial point process \cite{StoyanKendallMecke96}
$\Phi=\{x_{i}\}$, $i=1,2,\ldots n$, on $\mathcal{T}_{d}$.
Each element of $(G_{n})$ is based on the well-known disk graph model
\cite{Gilbert61,Penrose03}. In the following let $b_d(x,r)\equiv b(x,r)$ denote a Euclidean ball centered at $x\in\mathbb{R}^{d}$ and radius $r$, and $|b(x,r)|$ denote its volume.  

\subsubsection{\label{sub:Flat-Network-Model}Networks with Short-Range Communication}

In this case, $G_{n}$ is the $d-$dimensional disk graph parameterized
by the common \emph{communication range }$r$ of each node. The neighborhood of node $x_{i}\in\Phi$ that will be used for implementing
(\ref{eq:ScalarStateUpdate}) is 
\[ N_{x_i}(r)\triangleq\{x_j\in\Phi:\|x_{j}-x_{i}\|\leq r\}, \]

where $\|\cdot\|$ denotes the Euclidean norm. In this paper, we will always operate in the super-critical regime, i.e., $r=\omega(r_{c})$, where $r_{c}\triangleq(\frac{\log n}{n})^{1/d}$ to
ensure asymptotic connectivity and regularity of $(G_{n})$\cite{GK2000}. We label
this family of graphs as $G_{n}^{\,\mathrm{sh}}(r,d)\equiv G_{n}^{\,\mathrm{sh}}(r)$, and the update matrix by $\mathbf{W}\!_n^{\,\mathrm{sh}}$.
We refer to the points of $\Phi$ either by their location $x_i\in\mathbb{R}^d$ or by their index $i\in \mathbb{N}$.

\subsubsection{\label{sub:Tiered-Network-Model}Networks with both Short- and Selective
Long-Range Communication}

We start with a disk graph $G_{n}^{\,\mathrm{sh}}(r)$ and add long-range edges of length $s=\Theta(r^\gamma)$. The parameter $\gamma$ controls the distance over which long-range communication occurs: for a given $r$ a node can communicate with nodes farther away as $\gamma\rightarrow0$. We add the long edges as follows. 

For some $r, \eta>0$ and $0<\gamma<1$, tile the torus with hypercubes of side length $\eta r$. Let $c$ denote one of these hypercubes. Along each dimension $m=1,2,\ldots d$, let $c_{m}^{+}$ and $c_{m}^{-}$ denote the farthest hypercubes from $c$ that are less than distance $s/2$ away from $c$ along the $m^{\mathrm{th}}$
coordinate axis, the distance being measured in terms of the separation between their farthest edges. We call these hypercubes as the \emph{partner
hypercubes} of $c$. Figure \ref{fig:TieredNetworkFigure} illustrates the case of $d=2$. It is easy to see that from any vertex in $c$,
any vertex in $c_{m}^{+}$ and $c_{m}^{-}$ is at a distance of at most $\sqrt{(d-1)\eta^{2}r^{2}+s^{2}/4}\leq\frac{s}{\sqrt{2}}$ for a small enough $\eta$. 

Since $r=\omega(r_{c})$, every tile $c$ contains $n\eta^{2}r^{2}$
nodes a.a.s. Without loss of generality, let $x_{1}$ be one of these
nodes. Now add an edge between $x_{1}$ and every vertex in $c_{m}^{+},c_{m}^{-}$
for $m=1,2,\ldots,d$. Thus each of these nodes becomes a \emph{long-range
partner} of $x_{1}$. Repeat this procedure for every node in $\Phi$,
and count duplicate edges only once. Thus for $r=\omega(r_{c})$,
every node in every tile is additionally connected to $nr^{2}|b(0,1)|+2dn\eta^{2}r^{2}+o(1)$
nodes a.a.s., i.e., $G_{n}$ is regular asymptotically almost surely.
Hence an iteration of the form (\ref{eq:ConsVectUpdate}) on this
graph will converge to a uniform distribution a.a.s. We define the
resultant graph as $G_{n}^{\mathrm{l}}(r,s,d)\equiv G_{n}^{\mathrm{l}}(r,s)$ and the corresponding update matrix by $\mathbf{W}\!_n^{\,\mathrm{l}}$. 

\begin{figure}[hbt]
\centering
	\psfrag{etar}{$\eta r$}
	\psfrag{r}{$r$}
	\psfrag{leqsby2}{$\leq s/2$}
	\psfrag{gtsby2}{$>s/2$}
	\includegraphics[scale=0.7]{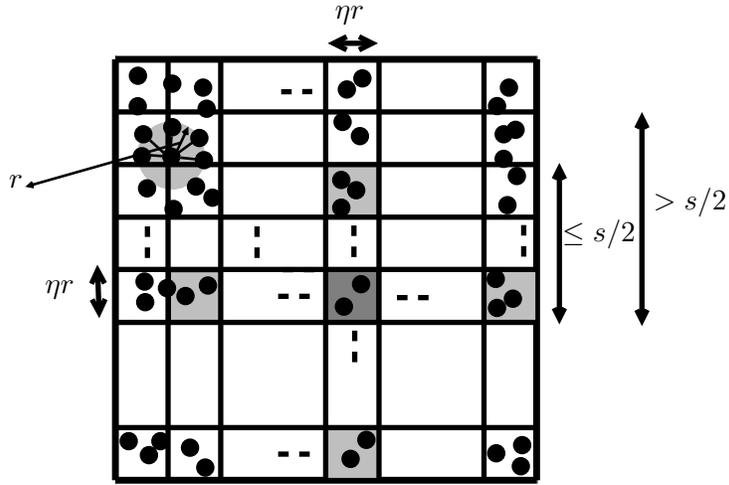}
	\caption{\label{fig:TieredNetworkFigure}An illustration of the geometric random graph models for $d=2$. The vertices are shown as black circles. In $G_n^{\,\mathrm{s}}(r)$, an edge exists between any two nodes iff they are at most at a distance $r$ (the communication range) away from each other. This is shown, for example, for the node at the center of the circle. $G_n^{\mathrm{l}}$ contains all edges in $G_n^{\,\mathrm{s}}$. Additionally each node communicates with its long-range partners. For example, for each node in the dark gray square, all nodes in the lightly shaded squares are long-range partners. These partner squares are chosen such that the distance between their farthest edges is less than $s/2$. Note that there are 4 such partner squares, two along each coordinate axis.}
\end{figure}

Notice that this model adds long edges selectively to each node; it is motivated by the observation that a small number of long
edges added to a graph can greatly increase its spectral gap, as is the case in small-world graphs (cf.~\cite[Chap. 14]{Durrett07}).
We have adapted this idea to a wireless network. Instead of adding a single additional edge to a node as is normally the case in abstract
graph-theoretic models, the inherent broadcast nature of the wireless channel allows a transmitter to broadcast its information to several
receivers that are in close proximity to one another with very little overhead. This allows multiple communication paths to form simultaneously.

We now describe the communication model, which is a well-accepted model in the study of wireless networks.

\subsection{Communication Model \label{sub:Comm-Assumptions}}

We make the following assumptions on the communication model: 
\begin{itemize}
\item All edges in $G_{n}^{\,\mathrm{sh}}$ and $G_{n}^{\mathrm{l}}$ are established
by wireless links that operate in the same frequency band (normalized to unit bandwidth).

\item Each node encodes its message in $K\gg1$ nats, such that there is
negligible quantization error. These messages are sent using a point-to-point
capacity-achieving AWGN channel code with SNR threshold $\beta$ (i.e.,
$R=\log(1+\beta)$). Transmissions are slotted with $K/R$ channel
uses allowed per slot.
\item There is no fading. 
The path-loss exponent $\alpha$ is greater than the dimension $d$
of the network, so that the interference remains finite a.s. as
the network size grows. 
\item A packet from node $i$ can be received at $j$ iff the Signal-to-Interference-Ratio
(SIR) at node $j$, $\text{SIR}_{ij}$, is greater than a known constant
$\beta>0$. Therefore for any sender $i$ and receiver $j$, the link
$i\rightarrow j$ will be in outage iff \begin{equation}
\frac{\|x_{j}-x_{i}\|^{-\alpha}}{\sum_{k\in\mathcal{S}\backslash\{i\}}\|x_{j}-x_{k}\|^{-\alpha}}<\beta.\label{eq:SINRModel}\end{equation}
 where $\mathcal{S}$ is the set of all senders that transmit in the
same slot as $i$. This is the well-known interference-limited physical model \cite{GK2000}\footnote{It is possible to derive our scaling results by including both noise and interference in the SINR model at the cost of making equations and derivations more cumbersome while distracting from the main message of the paper, which is the analysis of the performance with interference. Hence we focus on the interference-limited case.}.
\item The medium-access scheme is TDMA with spatial re-use. 
\end{itemize}
Thus the successful formation of each edge in a graph $G_{n}$ is
mapped to a successful link formation in each direction. Notice that
(\ref{eq:SINRModel}) models that fact that there is a limit to the
number of edges that can be formed simultaneously, and consequently
on the maximum rate at which a given message-passing graph can be
established. For a given TDMA protocol, the rate of topology formation
is thus determined by its \emph{schedule length} in time-slots. Since
we investigate networks in the scaling limit, we will investigate
the scaling properties of the fastest TDMA protocols that can establish
a given sequence of random graphs $(G_{n})$ (i.e., have the smallest
schedule length a.a.s.) 

%
%
\subsection{Quantifying the Effective Speed of Convergence}

Note that the mixing time, which is a function of the update matrix $\mathbf{W}\!_n$, the smallest \emph{number of iterations} to converge to an $\epsilon-$ball around the average consensus point. This is different from the \emph{time taken} to taken to converge to this ball with a finite rate of topology formation in each iteration. For example, in Fig. \ref{fig:SixNodeExample}, due to interference constraints, the shortest schedule to construct this topology has at least three time slots. Hence message-passing iterations using this topology can occur no faster than once in every three time slots. 

Thus for a topology $G_n$ and an update matrix $\mathbf{W}\!_n$, the \emph{smallest} effective time to converge is the \emph{product} of the mixing time $T_{\mathrm{mix}}(\epsilon;\mathbf{W}\!_n)$ of a topology and the length $T^{*}(G_n,\beta)$ of the \emph{shortest} TDMA schedule that constructs the topology in each iteration. We call this the \emph{Slot Mixing Time}. We formally state it below for future reference:

\begin{defn}
\label{def:SlotMixingTime}(Slot Mixing Time) The Slot Mixing Time $T_{\mathrm{slots}}(G_n)\equiv T_{\mathrm{slots}}(G_n,\mathbf{W}\!_n,\beta,\epsilon)$ is defined as the product
\[ T_{\mathrm{slots}}(G_n)\triangleq T_{\mathrm{mix}}(\epsilon;\mathbf{W}\!_n)\cdot T^{*}(G_n,\beta), \]
where $T_{\mathrm{mix}}(\epsilon;\mathbf{W}\!_n)$ is the $\epsilon-$mixing time of iterations using a message-passing graph $G_n$ and an update matrix $\mathbf{W}\!_n$ and $T^{*}(G_n,\beta)$ is the length of the shortest TDMA schedule that constructs $G_n$ in time slots.
\end{defn}

Notice that in general $T_{\mathrm{slots}}(G_n)$ depends on the realization of the random graph $G_n$. We will analyze the scaling of $T_{\mathrm{slots}}(G_n)$ for the families of random geometric graphs described in Section \ref{sub:Network-Models}.

\subsection{Asymptotic Behavior}

From Sections \ref{sub:Consensus-RandomWalks} and \ref{sub:Comm-Assumptions}
we notice that the problem involves:
\begin{itemize}
\item The network size $n$.
\item The short link distance $r$. 
\item The parameter $\gamma$ that controls the length of long links.
\end{itemize}
We will study the mixing time in an interference-limited network in
the regime $n\rightarrow\infty$. 

\section{\label{sec:Cvg-Flat-Networks}Convergence in Networks with Small
Communication Range}

\subsection{Characterizing the Spectral Gap\label{sub:SRMixingTimeIterations}}

The spectral gap for the disk graph is known to be $\Theta(r^{2})$, independent of network dimension~\cite{Boyd06}. 
Using Cheeger's Inequality (Theorem \ref{thm:MixingTimeBounds}), it was shown that the mixing
time of the \emph{fastest mixing} reversible random walk with a uniform
distribution on $G_{n}^{\,\mathrm{sh}}(r)$, for polynomial scaling $\epsilon=1/n^{\delta}$,
$\delta>0$ scales as\begin{equation}
T_{\mathrm{mix}}(\mathbf{W}\!_n^{\,\mathrm{sh}})=\Theta(r^{-2}\log n).\label{eq:MixingTimeRandomWalk}\end{equation}

It was also shown therein that the mixing time for the natural random walk on $G_n^{\,\mathrm{sh}}$ is also $\Theta(r^{-2}\log n)$. We will now use combine the scaling law for the mixing time with the fastest rate of topology formation implied by the communication
model in Section \ref{sub:Comm-Assumptions}.

\subsection{Interference-Limited Topology Formation\label{sec:ShortRangeInterference}}

We now prove two results that follow from the assumptions made in
Section \ref{sub:Comm-Assumptions}. 

\begin{prop} \label{pro:ND-Transmission-Steps-LB} Consider a system of $n$ nodes on a $d-$dimensional torus with a short-range communication range $r$, that communicate using point-to-point codes with SINR threshold $\beta$, with $\alpha>d$ being the path-loss exponent. Assuming the short-range network model in Section \ref{sub:Flat-Network-Model} and the communication model described in \ref{sub:Comm-Assumptions}, the length of the shortest TDMA schedule that constructs $G_{n}^{\,\mathrm{sh}}$ has no fewer than $C_{1}nr^{d}\beta^{d/\alpha}$ slots a.a.s., for some positive
constant $C_{1}$.\end{prop} \begin{proof} Let $\mathcal{S}$ be
the set of concurrent transmitters at any given time. Suppose node
$j$ is an intended receiver of a transmitter $i\in\mathcal{S}$.
Then $i$'s message is decoded correctly iff (\ref{eq:SINRModel})
is satisfied. Thus for all $k\in\mathcal{S}\backslash\{i\}$, \begin{eqnarray}
\|x_{j}-x_{k}\| & \geq & \beta^{1/\alpha}\|x_{j}-x_{i}\|.\label{eq:Lower-Bound-SIR-condition}\end{eqnarray}
Clearly this is true even for the \emph{farthest }intended receiver. It is easy to show that such a receiver lies a.a.s. in a ring of inner
radius $s(1-\delta^{''})$ for some fixed $\delta^{''}>0$. We thus conclude
$\|x_{k}-x_{j}\|\geq r(1-\delta^{''})\beta^{1/\alpha}\triangleq r_{\min}$
a.a.s.

This suggests that any TDMA protocol allowing $i$ to pass a message
to its farthest node $j$ needs to set up a guard zone of radius no
smaller than $r_{\min}$ around $j$. Since every node inside this
guard zone must transmit at least once to form the required message
passing graph, any TDMA protocol that constructs the message passing
graph $G_{n}^{\,\mathrm{sh}}$ requires least $\sum_{x\in\Phi}\mathbf{1}_{x\in \Phi\cap b(0,r_{\min})}$
slots. Here the indicator $\mathbf{1}_{x\in\Phi\cap b(0,r_{\min})}$ is used to indicate the existence of the point $x\in\Phi$ inside the ball $b(0,r_{\min})$. The summation is over all points $x\in\Phi$. 

For $r=\omega(r_{c})$, each such ball has $n|b(0,r_{\min})|=nr^{d}\beta^{d/\alpha}(1-\delta)^{d}|b(0,1)|+o(1)\geq C_{1}nr^{d}\beta^{d/\alpha}$ a.a.s.,
where $C_{1}=0.5(1-\delta)^{d}|b(0,1)|$. \end{proof} 

\begin{prop} \label{pro:2D-Transmission-Steps-UB}Consider the network
model in Section \ref{sub:Flat-Network-Model} and the communication
model described in \ref{sub:Comm-Assumptions}. The length of the
shortest TDMA schedule that constructs $G_{n}^{\,\mathrm{sh}}$ has at
most $C_{2}nr^{d}\beta^{d/\alpha}$ slots a.a.s., for some positive
constant $C_{2}$.\end{prop} 
\begin{proof} The proof involves construction
of a feasible TDMA schedule whose length is $C_{2}nr^{d}\beta^{d/\alpha}$
Let $x\triangleq\theta r$ for some fixed $\theta>1$. Consider
the lattice $\mathbb{L}$ that consists of points on the scaled integer
lattice $x\mathbb{Z}^{2}$ that also lie on the torus. In other words,
$\mathbb{L}=x\mathbb{Z}^{2}\cap\mathcal{T}_{2}(n)$. Partition $\mathbb{L}$
into sublattices as follows: 
\begin{itemize}
\item $\mathbb{L}_{00}\triangleq\{(ix,jx)\in\mathbb{L}\,:\, i\,\text{and}\, j\,\text{are\ensuremath{\,}even}\}$ 
\item $\mathbb{L}_{01}\triangleq\{(ix,jx)\in\mathbb{L}\,:\, i\,\text{even},\, j\,\text{odd}\}$ 
\item $\mathbb{L}_{10}\triangleq\{(ix,jx)\in\mathbb{L}\,:\, i\,\text{odd},\, j\,\text{even}\}$ 
\item $\mathbb{L}_{11}\triangleq\{(ix,jx)\in\mathbb{L}\,:\, i\,\text{and}\, j\,\text{are\ensuremath{\,}odd}\}$ 
\end{itemize}
With each lattice site $p\in\mathbb{L}$ one can associate the tile
$\tau_{p}=p+[0,x]^{2}$ that lies within the torus $\mathcal{T}_{2}(n)$.
Denote by $\mathbb{T}_{ij}$ the set of such tiles associated with
each of the points in $\mathbb{L}_{ij}$, $i,j=0,1$. For example,
$\mathbb{T}_{00}\triangleq\{\tau_{p}:\, p\in\mathbb{L}_{00}\}$. Thus
$\{\mathbb{T}_{ij}\}$ partition the torus $\mathcal{T}_{2}(n)$.

The idea behind such a partition is to enable spatial re-use. Consider
the following four-phase MAC protocol consisting of phases 00, 01,
10, 11. In phase $ij$ at most one node from each tile in $\mathbb{T}_{ij}$
is allowed to transmit. The protocol ensures that each node transmits
exactly once.

The next step is to show that this protocol provides the desired connectivity
to each node every $C_{2}nr^{2}\beta^{2/\alpha}$ time slots for some
positive $C_{2}$. To this end, we first show that the interference
at each intended receiver is bounded from above and can be made smaller
than any $\beta>0$ by a suitable choice of $\theta$.

Consider one such transmission in phase 00. Let $\mathcal{S}\subset\mathbb{T}_{00}\cap V_{n}$
be the set of all transmitters. Consider a transmitting node $i$
in tile $\tau_{p}$ where $p=(0,0)$, i.e., a tile at the origin.
To remain feasible, the protocol must satisfy (\ref{eq:SINRModel})
for each successful link. For any $i,j,k$, it is clear that \begin{eqnarray*}
\|x_{k}-x_{j}\| & = & \|x_{k}-x_{i}-(x_{j}-x_{i})\|\\
 & \geq & \|x_{k}-x_{i}\|-\|x_{j}-x_{i}\|\\
 & \geq & \|x_{k}-x_{i}\|-r,\end{eqnarray*}
since $\|x_{j}-x_{i}\|\leq r$. Therefore for a transmitter at $x_{i}$,
the interference power at any intended receiver at $x_{j}$ can be
upper bounded as \begin{equation}
\sum_{k\in\mathcal{S}\backslash\{i\}}\|x_{k}-x_{j}\|^{-\alpha}\leq\sum_{k\in\mathcal{S}\backslash\{i\}}\left(\|x_{k}-x_{i}\|-r\right)^{-\alpha},\label{eq:Interference-Power-UB}\end{equation}
 where the right hand side is independent of $j$. By the design of
the protocol, an interferer $k$ for any intended receiver of the
message from $i$ must lie in a tile distinct from $\tau_{(0,0)}$.
Moreover, such a tile should lie within $\mathbb{T}_{00}$; thus the
protocol imposes a lower bound on the minimum distance between any
two concurrent transmitters. Using geometrical arguments (see Figure
\ref{fig:2D-Tiling}), 
the right hand side of (\ref{eq:Interference-Power-UB}) is upper
bounded as 
\begin{align}
 & \sum_{k\in\mathcal{S}\backslash\{i\}}\left(\|x_{k}-x_{i}\|-r\right)^{-\alpha}\nonumber \\
 & \leq\sum_{l=1}^{\infty}8l\left((2l-1)\theta r-r\right)^{-\alpha}\nonumber \\
 & =8r^{-\alpha}\sum_{l=1}^{\infty}l\left((2l-1)\theta-1\right)^{-\alpha}\nonumber \\
 & \leq8r^{-\alpha}\left(((\theta-1)^{-\alpha}+\sum_{l=2}^{\infty}l((2l-1)\theta-\theta)^{-\alpha}\right)\nonumber \\
 & =8r^{-\alpha}\left((\theta-1)^{-\alpha}+2^{-\alpha}\theta^{-\alpha}\sum_{l=2}^{\infty}l(l-1)^{-\alpha}\right)\nonumber \\
 & \leq\xi r^{-\alpha}(\theta-1)^{-\alpha},\label{eq:UB-on-Interference-Power-UB}\end{align}
\begin{figure}
\centering 
\includegraphics[width=3.2in]{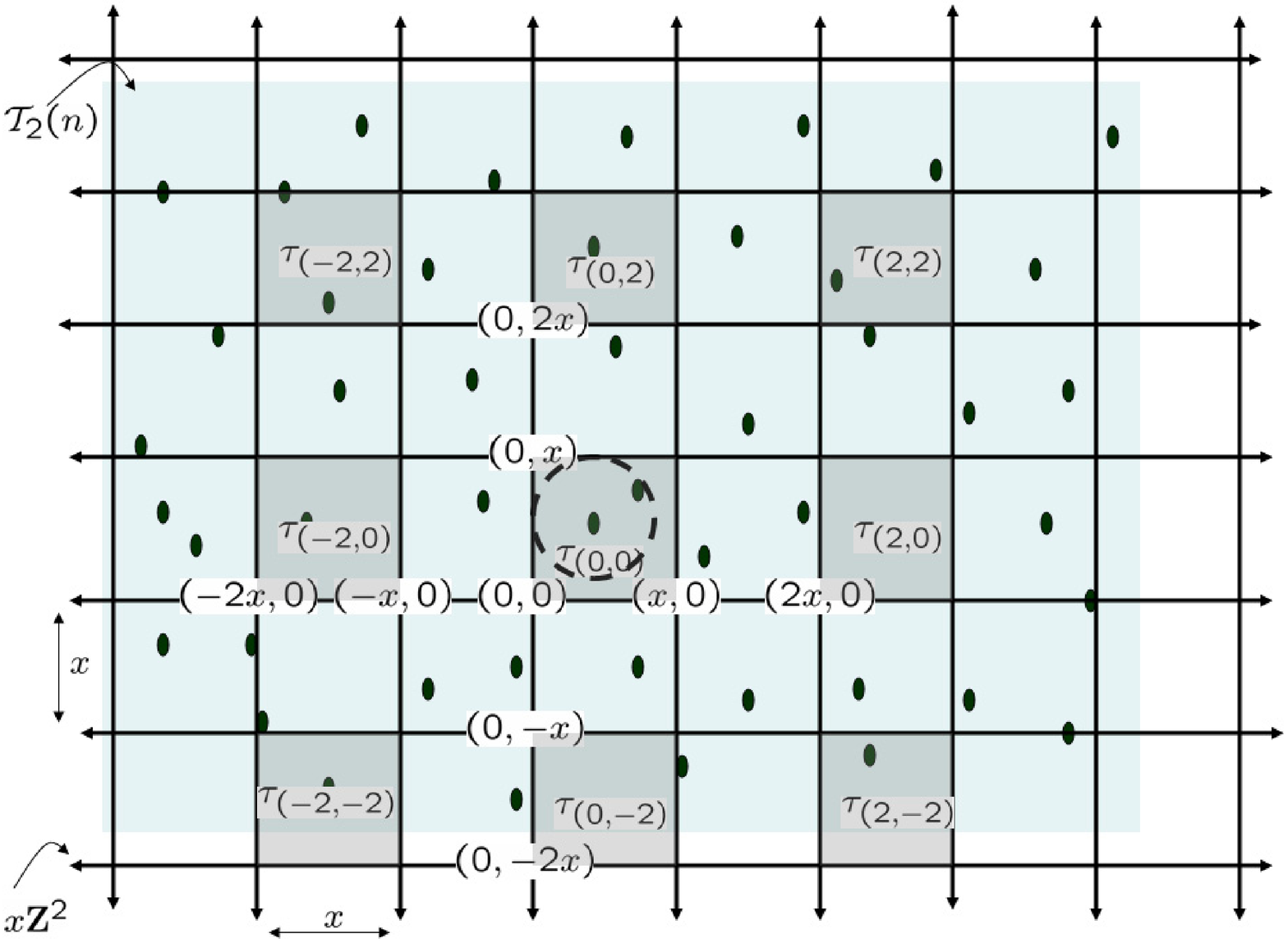} 

\caption{Geometric reasoning underlying the proof of Proposition \ref{pro:2D-Transmission-Steps-UB}.
The location of a typical transmitter in $\tau_{0,0}$ and one of
its intended receivers is shown. The nearest interferers and their
respective tiles are $\tau_{2,0},\tau_{2,2},\tau_{0,2},\tau_{-2,2},\tau_{-2,0},\tau_{-2,-2},\tau_{0,-2},\tau_{2,-2}$.
The signal power from any of one these interferers at the intended
receiver is no larger than that received from the closest interferer
allowed by the protocol. The protocol ensures that this nearest distance
is no smaller than $x=\theta r$.}

\label{fig:2D-Tiling} 
\end{figure}

for some fixed $\xi>0$, since the sum converges for $\alpha>2$ (in general, for $\alpha>d$, as assumed in the communication model). The SIR condition (\ref{eq:SINRModel})
is guaranteed to be satisfied at every intended $j$, if $\theta$
is chosen such that \begin{eqnarray*}
\frac{r^{-\alpha}}{\xi r^{-\alpha}(\theta-1)^{-\alpha}} & \geq & \beta\\
\implies\theta & \geq & 1+(\xi\beta)^{\frac{1}{\alpha}}.\end{eqnarray*}
 For a suitable choice of $\xi$, we can set $\theta=10(\xi\beta)^{\frac{1}{\alpha}}$.

For $s=\omega(s_{c})$, the number of nodes in each tile is $nx^{2}+o(1)$
a.a.s. Hence as $n\rightarrow\infty$, the protocol constructed requires
$4nx^{2}+o(1)\leq C_{2}ns^{2}\beta^{2/\alpha}$ 
transmissions almost surely to establish the necessary connectivity
to each node in the network, where $C_{2}\geq400\xi^{2/\alpha}$.
By optimality, the number of slots $T^{*}$ in the shortest TDMA schedule
cannot exceed this number. 
\end{proof}

The results from Propositions \ref{pro:ND-Transmission-Steps-LB}
and \ref{pro:2D-Transmission-Steps-UB} lead to the following corollary.
\begin{cor}\label{cor: ShortestTDMAschedule} If $T^{*}(G_n^{\,\mathrm{sh}}(r),\beta)$
denotes the length of the shortest TDMA schedule, then as $n\rightarrow\infty$, a.a.s.:
\begin{enumerate}
\item For fixed $\beta$, $T^{*}(G_n^{\,\mathrm{sh}},\beta)=\Theta(nr^{d})$.
\item When $\beta\equiv\beta(n)=\Omega(1)$, $T^{*}(G_n^{\,\mathrm{sh}},\beta)=\Omega(nr^{d}\{\beta(n)\}^{\frac{d}{\alpha}})$.
\end{enumerate}
\end{cor} \begin{proof} Claim 1 is evident from the results of Propositions
\ref{pro:ND-Transmission-Steps-LB} and \ref{pro:2D-Transmission-Steps-UB}.

For some constants $C_{1}$ and $C_{2}$, we have from Propositions
\ref{pro:ND-Transmission-Steps-LB} and \ref{pro:2D-Transmission-Steps-UB},
a.a.s. for large $n$ and a fixed $\beta$, \[
C_{1}nr^{d}\beta^{d/\alpha}\leq T^{*}(G_n^{\,\mathrm{sh}},\beta)\leq C_{2}nr^{d}\beta^{d/\alpha}.\]

Since $C_{1}$ (but not $C_{2}$) is independent of $\beta$, we can write for $n\rightarrow\infty$, when $\beta\equiv\beta(n)=\Omega(1)$ $T^{*}(G_n^{\,\mathrm{sh}})=\Omega(nr^{d}\beta^{\frac{d}{\alpha}}(n))$.\end{proof}

If all nodes had independent point-to-point channels between one another,
the rate of topology formation would be $\Theta(1)$. For a wireless channel, however, Corollary \ref{cor: ShortestTDMAschedule} suggests that it requires $\Theta(1/nr^d)$ even with optimum spatial re-use. Thus better-connected disk graphs are penalized by a smaller rate of topology formation. We combine the mixing time result (\ref{eq:MixingTimeRandomWalk}) to examine the scaling law for the effective \emph{time} necessary for convergence in the next section.


\subsection{Rate of Convergence\label{sec:Convergence-In-Dense-Networks}}

\subsubsection{Slot Mixing Time}

We now analyze the asymptotic convergence behavior of the distributed averaging algorithm (\ref{eq:ConsVectUpdate}) in a dense network
as $n\rightarrow\infty$. From the earlier sections, we know the scaling laws for this regime for:
\begin{enumerate}
	\item The \emph{number of iterations} necessary to a.s. reach an $\epsilon-$ball (from (\ref{eq:EpsCvgTime})).
	\item The \emph{shortest TDMA schedule length} to a.s. realize $G_{n}^{\,\mathrm{sh}}$ in each iteration (from Corollary \ref{cor: ShortestTDMAschedule}).

\end{enumerate}
Thus from Definition \ref{def:SlotMixingTime}, for fixed $\beta$, the Slot Mixing Time scales as
\begin{equation}
T_{\mathrm{slots}}(G_n^{\,\mathrm{sh}})\triangleq T_{\mathrm{mix}}(\mathbf{W}\!_n^{\,\mathrm{sh}})\cdot T^{*}(G_n^{\,\mathrm{sh}},\beta)=\Theta(nr^{d-2}\log n)\label{eq:FinalMixingTimeSlots}\end{equation}
slots a.a.s., for $\epsilon=1/n^{\delta}$.

From Proposition \ref{pro:ND-Transmission-Steps-LB} and the Gaussian signaling assumption, when we also allow $\beta$ to depend on $n$ such that $\beta(n)=\Omega(1)$, the \emph{time} to reach this ball scales as \[
\Omega\left(nr^{d-2}\frac{e^{R(n)d/\alpha}}{R(n)}\log n\right)\qquad\mathrm{a.a.s.}\]

where $R(n)\equiv\log(1+\beta(n))$.

\subsubsection{Choice of Communication Range}

For a fixed $\beta$ the mixing time in (\ref{eq:FinalMixingTimeSlots}) scales polynomially in $r$ for $d>1$. Interestingly, for $d=1$, the time slots to mix scales as the inverse of $r$. This suggests that increasing $r$ can improve the rate of convergence. For $d=2$, however, this quantity scales \emph{independently} of $r$, suggesting that these two effects exactly cancel each other, a rather non-intuitive result. For higher dimensions, the scaling law has a positive exponent in $r$---implying that the increasing $r$ can actually slow down mixing.

This dependence on network dimension can be understood as follows. If the network is one-dimensional, although a transmitter is an isotropic
radiator, its effect on the network is seen only along the line $[0,1]$. Although the throughput provided by the optimal TDMA protocol only
scales as $\Theta(n^{-1}r^{-1})$ for a given $\beta$ from Corollary \ref{cor: ShortestTDMAschedule}, the spectral gap scales as $\Theta(r^{-2})$,
offsetting this loss. In $d-$dimensions, however, while the the fastest rate of topology formation scales as $\Theta(n^{-1}r^{-d})$, the spectral gap only
scales as $\Theta(r^{-2})$. As a result, improving spatial re-use can become more important than increasing connectivity.

\subsubsection{Effect of Increasing Transmission Rate\label{sub:ChoiceTransmRate}}

On the one hand, higher transmission rate reduces the packet transmission time; on the other, it also restricts spatial re-use. Clearly the benefit
of smaller packet transmission times can be outweighed by reduced spatial re-use for large rates $R$.

\section{\label{sec:Cvg-Tiered-Networks}Convergence in Networks with Selective Long-Range Connectivity}

\subsection{\label{sub:SpecGapTiered}Scaling of the Spectral Gap}


To derive the scaling law for the mixing time, we need to find the scaling of the spectral gap of $G_{n}^{\mathrm{l}}$. As we will see, deriving the 
scaling law for the conductance of $G_{n}^{\mathrm{l}}$ is sufficient to establish the scaling of the spectral gap.

\begin{prop}\label{prop:ConductanceBound}The conductance of $G_{n}^{\mathrm{l}}$ with edge weights determined by $\mathbf{W}\!_n^{\,\mathrm{l}}$ 
is $\Theta(r^{\gamma})$ a.a.s., for $d=1,2,\ldots$\end{prop}

\begin{IEEEproof}We adopt a modified version of the proof in~\cite{Avin07}. 
From (\ref{eq:ConductanceDef}) we know that \[
h=\min_{S\subset\Omega,\pi^{*}(S)\leq1/2}\frac{Q(S,\bar{S})}{\pi^{*}(S)}.\]
By the symmetry in $G_{n}^{\mathrm{l}}$ induced by the construction
in Section \ref{sub:Tiered-Network-Model}, it can be shown using
arguments similar to~\cite[Appendix G]{Avin07} that the minimum occurs for $\pi^{*}(S)=1/2$, and that the minimizing cut
$(S,\bar{S})$ is a hyperplane dividing the torus into two halves.
Without loss of generality, define $S\triangleq\Phi\cap\{[0,1/2)\times[0,1]\}$. 

\begin{figure}[hbt]
\centering
	\psfrag{numrows}{$\begin{matrix} \Theta(1/r) \\ \text{rows} \end{matrix}$}
	\psfrag{numsqrs}{$\Theta(r^\gamma/r)$}
	\psfrag{s}{$\Theta(r)$}
	\psfrag{S}{$S$}
	\psfrag{Sbar}{$\bar{S}$}
	\psfrag{Cut}{$\text{Cut}(S,\bar{S})$}
	\includegraphics[scale=0.7]{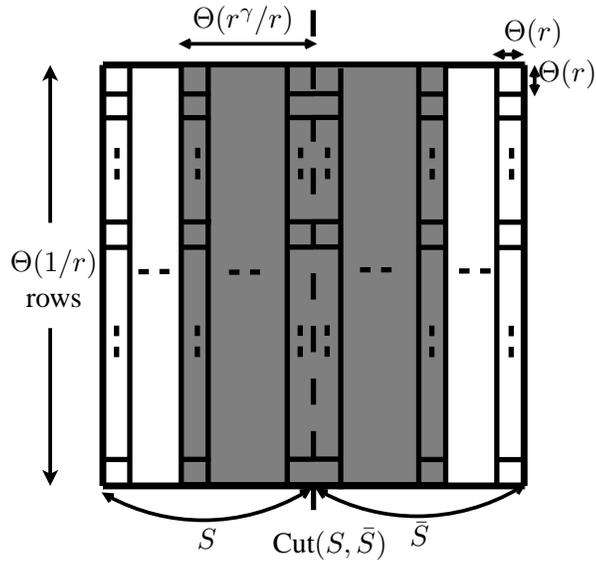}
	\caption{\label{fig:ConductanceBound}The geometry behind the proof of Proposition \ref{prop:ConductanceBound} for $d=2$. The tiling used for the construction of $G_n^{\mathrm{l}}$ is overlaid. By the symmetry induced by the construction, the set $S\subset V_n$ for which $Q(S,\bar{S})/\pi^{*}(S)$ is minimized corresponds to the left-half of the torus as labeled (it can be argued that this set will have the smallest weighted flow for a given frequency of steady-state occupancy). Since the stationary distribution for this set is $1/2$, finding the scaling law for the number of edges that traverse the cut is sufficient to provide a corresponding scaling result for the conductance. For the short-range communication graph $G_n^{\,\mathrm{s}}$ (i.e., the disk graph whose edge length is $\mathcal{O}(r)$) only nodes from a finite number of squares from the tiling in either direction from the cut contribute to these edges. For long edges of length $\Theta(r^\gamma)$, a positive fraction of the nodes from $\Theta(r^\gamma/r)$ squares on either side will contribute to these edges. Since there are $\Theta(1/r)$ such rows of squares, the proof lies in finding the scaling law for the number of edges that traverse the cut.}
\end{figure}

Also for the natural random walk, each edge weight is $\frac{1}{d_{i}}=\Theta\left(\frac{1}{nr^{2}}\right)$
(for $d$ dimensions, $\Theta(n^{-1}r^{-d})$), and the equilibrium
distribution is $\Theta(\frac{1}{n})$. It is thus sufficient to count
the number of edges traversing this cut. The number of short
edges was shown in~\cite{Avin07} to be $\Theta(n^{2}r^{3})$
(for $d$ dimensions $\Theta(n^{2}r^{d+1})$). Observe that every
node in a square of side $\eta r$ has $4n\eta^{2}r^{2}$ long-range
partners. One quarter of these edges traverse the cut $(S,\bar{S})$; hence the potential number of long edges that can traverse the cut from a given square is $n\eta^{2}r^{2}\times n\eta^{2}r^{2}=\Theta(n^{2}r^{4})$.
Since each edge has length at least $s/2-2\eta r=\Theta(r^{\gamma})$ (since $s=\Theta(r^{\gamma})$ and $0<\gamma<1$), which is at most $s$, it is clear that 
$\Theta(r^{\gamma-1})$ squares from the cut will contribute to the edges that traverse the cut (see Fig. \ref{fig:ConductanceBound}). Multiplying this result by the number of rows $\Theta(r^{-1})$ of such squares, the total number of long edges traversing the cut will be $\Theta(n^{2}r^{4}\times r^{\gamma-1}\times r^{-1})=\Theta(n^{2}r^{2+\gamma})$ (for general $d$, $\Theta(n^{2}r^{2d}\times r^{\gamma-1}\times r^{-d+1})=\Theta(n^{2}r^{d+\gamma})$).
Counting both the short and long edges, we have in $d$ dimensions, \begin{eqnarray*}
Q(S,\bar{S}) & = & \Theta\left(\frac{n^{2}r^{d+1}+n^{2}r^{d+\gamma}}{n^{2}r^{d}}\right)\\
 & = & \Theta(r^{\gamma}),\end{eqnarray*}
since $\gamma<1$. 

Notice that if a node were allowed to have only a finite number of long-range partners, the contribution of long-edges towards conductance is smaller, without significant interference-reducing benefits. We elaborate on this point in Section \ref{sub:LengthAndDensityLRLinks}. \end{IEEEproof}

We can infer the following from the above result:

\begin{cor}\label{cor:SpecGapBounds}The spectral gap of $G_{n}^{\mathrm{l}}$
is $\Omega(r^{2\gamma})$ and $\mathcal{O}(r^{\gamma})$.\end{cor}

\begin{IEEEproof} From the lower bound in Theorem \ref{thm:CheegersInequality},
we have $1-\mu_{2}=\Omega(r^{2\gamma})$. From the upper bound from the same theorem, we have $1-\mu_{2}=\mathcal{O}(r^{\gamma})$.\end{IEEEproof}

As noted in Section \ref{sub:Tiered-Network-Model}, the distance
between any two (graph-theoretic) neighbors is no more than $s/\sqrt{2}$.
Thus every edge in $G_{n}^{\mathrm{l}}(r,s,d)$ is also present in
the disk graph $G_{n}^{\,\mathrm{sh}}(s/2)$, i.e., $G_{n}^{\mathrm{l}}(r,s,d)\subset G_{n}^{\,\mathrm{sh}}(s/2)$.
Hence a reversible random walk on $G_{n}^{\mathrm{l}}$ with a uniform
equilibrium distribution can mix \emph{no faster than the fastest
mixing such random walk on $G_{n}^{\,\mathrm{sh}}(s/2)$}. This key observation
allows us to use a known result that follows from ~\cite[Thm. 8]{Boyd06}: 

\begin{thm}\label{prop:BestRandomWalk}The spectral gap corresponding to the transition probability matrix of the fastest
mixing reversible random walk on $G_{n}^{\,\mathrm{sh}}(r^{\gamma})$
with a uniform equilibrium distribution is $\Theta(r^{2\gamma})$
a.a.s.\end{thm}

Since mixing time decreases with spectral gap, from Theorem \ref{prop:BestRandomWalk}
we conclude that the spectral gap of $G_{n}^{\mathrm{l}}$ is $\mathcal{O}(r^{2\gamma})$.
But we know from Corollary \ref{cor:SpecGapBounds} that this gap
is also $\Omega(r^{2\gamma})$. Thus we conclude that the spectral
gap of $G_{n}^{\mathrm{l}}$ is $\Theta(r^{2\gamma})$, which is formally
stated as a theorem:

\begin{thm}\label{prop:SpecGapMainResult} The spectral gap of the
natural random walk on $G_{n}^{\mathrm{l}}$ is $\Theta(r^{2\gamma})$.
\end{thm}

This result suggests that the improvement in spectral gap from an increased
communication radius from $r$ to $r^{\gamma}$ can also be achieved
(in the scaling sense) by allowing each node to communicate with a
selected number of nodes at a distance $\Theta(r^{\gamma})$. 

However, as we shall discuss in the next section, such connectivity
comes at a price of a lowered rate of topology formation. We find
that this loss (as measured by the shortest TDMA schedule length) must be no smaller than the number of nodes in the largest exclusion zone created in the network.
Since the longest link distance in both the disk graph $G_{n}^{\,\mathrm{sh}}(s/2)$ and $G_{n}^{\mathrm{l}}$ are of the same order,
the similarity in the expressions for the spectral gap scaling law suggests that we should expect the same dependence on network dimension as
in (\ref{eq:FinalMixingTimeSlots}).

\subsection{\label{sec:HubConvergence}Convergence with Interference}

We will derive bounds for the shortest feasible TDMA schedule for $G_{n}^{\mathrm{l}}$. In the spirit of the earlier proofs, the lower
bound follows from the feasibility constraint (i.e., the schedule
constructs the desired message passing graph while satisfying the
SINR constraint), while the upper bound is found by bounding the length
of the optimum schedule by that of a specific feasible schedule. These
results are presented in the following.

\begin{prop}\label{prop:ScheduleLB}For a given $\beta$, a feasible schedule for $G_n^{\mathrm{l}}$ has $C_{3}nr^{\gamma d}\beta^{d/\alpha}$ slots a.a.s. for some positive constant $C_{3}$. Furthermore, for a given $\beta$, this length scales as $\Omega(nr^{\gamma d})$ slots a.a.s. \end{prop}

\begin{IEEEproof}We prove this result for $d=2$; the proof for $d\neq2$ is similar. From the system model, it is clear that a TDMA protocol that constructs $G_{n}^{\mathrm{l}}$ must form at least one link of distance at least $s/2-2\eta r$. Since $s=\Theta(r^\gamma)$ (i.e., $s$ scales ``much slower'' than $r$), at large enough $n$, the protocol must create an exclusion zone of radius of at least $s/4$ in the network at least once. All nodes within this exclusion zone must transmit at least once. But $s=\omega(r_c^{\gamma})$, which implies we operate the supercritical regime. From a similar argument as in Proposition \ref{pro:ND-Transmission-Steps-LB} we can assert that any feasible TDMA protocol must have at least $C_{3}ns^{d}\beta^{d/\alpha}$ slots where $C_{3}$ is a positive constant. The scaling law for this length follows from the scaling of $s$ with $r$.\end{IEEEproof}

\begin{prop}\label{prop:ScheduleUB}For a given $\beta$, the length of the shortest feasible
schedule for $G_n^{\mathrm{l}}$ is no more than $C_{4}(ns^{d}\beta^{d/\alpha})$ slots a.a.s., for some positive constant $C_{4}$. For a given $\beta$, this upper bound scales as $\mathcal{O}(nr^{d})$ a.a.s.\end{prop}

\begin{IEEEproof}Consider any TDMA protocol that allows each node
to communicate with every node within a distance $s$. Clearly this
protocol will also construct $G_{n}^{\mathrm{l}}$ and is hence feasible.
As in Proposition \ref{pro:2D-Transmission-Steps-UB}, we construct
such a four-phase (for $d=2$, in general a $2d$ phase) TDMA protocol
that operates on a tiling of the torus with squares of side $\Theta(s)$.
Using an argument similar to Proposition \ref{pro:2D-Transmission-Steps-UB},
it is clear that the spatial re-use can be adjusted to construct the
graph in $C_{4}ns^{d}\beta^{d/\alpha}$ slots a.a.s. for some constant $C_4>0$. Using $s=\Theta(r^\gamma)$ we get the scaling law. \end{IEEEproof}

\begin{cor}\label{cor:ScheduleLongRangeScaling}As $n\rightarrow\infty$, the shortest feasible schedule for $G_n^{\mathrm{l}}$ has $T^{*}(G_n^{\mathrm{l}},\beta)=\Theta(nr^{\gamma d})$ slots a.a.s., for a fixed $\beta$. If we also let $\beta=\beta(n)=\Omega(1)$, $T^{*}(G_n^{\mathrm{l}},\beta(n))=\Omega(nr^{\gamma d}\{\beta(n)\}^{\frac{d}{\alpha}})$.\end{cor}

\begin{IEEEproof}Follows from Propositions \ref{prop:ScheduleLB}
and \ref{prop:ScheduleUB}.\end{IEEEproof}

\subsection{Rate of Convergence with Sparse Long-Range Connectivity\label{sec:Convergence-With-Sparse-Long-Range-Links}}

We repeat the analysis in Section~\ref{sec:Convergence-In-Dense-Networks} to study the benefit of sparse long-range connectivity for a large number of nodes. From this analysis, we derive a result analogous to (\ref{eq:FinalMixingTimeSlots}) for the long-range model. We use this result to discuss the impact of increased communication range. 

\subsubsection{Slot Mixing Time}

From Theorem \ref{prop:SpecGapMainResult}, the spectral gap of $G_n^{\mathrm{l}}$ scales as $\Theta(r^{2\gamma})$. Consequently,~from the mixing time bounds in Theorem \ref{thm:MixingTimeBounds}, we conclude that the mixing time with $\mathbf{W}\!_n^{\,\mathrm{l}}$ scales as 
\begin{eqnarray}
 T_{\mathrm{mix}}(\mathbf{W}\!_n^{\,\mathrm{l}}) &=& \Theta(r^{2\gamma}\log n)\qquad\mathrm{a.a.s.}\label{eq:MixingTimeLongRange}
\end{eqnarray}
iterations for $\epsilon=1/n^\delta$. On the other hand, from Corollary \ref{cor:ScheduleLongRangeScaling} the shortest TDMA schedule that realizes $G_n^{\mathrm{l}}$ scales as $\Theta(nr^{d\gamma})$ slots.

Multiplying $T_{\mathrm{mix}}(\mathbf{W}\!_n^{\,\mathrm{l}})$ and $T^{*}(G_n^{\mathrm{l}},\beta)$ we obtain a scaling law analogous to (\ref{eq:FinalMixingTimeSlots}) for a network with sparse long links. We state this result as a proposition:

\begin{prop}\label{pro:MixingTimeSlotsLongRange}As $n\rightarrow\infty$, for $\epsilon=1/n^{\delta}$ and with the shortest feasible
TDMA schedule, the slot mixing time of natural random walks on a sequence of random graphs $(G_n^{\mathrm{l}})$ on a $d-$dimensional torus scales as\begin{eqnarray}
T_{\mathrm{slots}}(G_n^{\mathrm{l}})\triangleq T_{\mathrm{mix}}(\mathbf{W}\!_n^{\,\mathrm{l}})\cdot T^{*}(G_n^{\mathrm{l}},\beta) & = & \Theta\left(nr^{(d-2)\gamma}\log n \right)\quad\mathrm{a.a.s.}\label{eq:MixingSlotsLongRange}
\end{eqnarray}
where $r$ is the short range communication radius, long links are $\Theta(r^\gamma)$ for some $0<\gamma<1$, and nodes use point-to-point capacity-achieving AWGN channel codes with SNR threshold $\beta$. 
 \end{prop}

From Proposition \ref{pro:ND-Transmission-Steps-LB} and the Gaussian signaling assumption in Section \ref{sub:Comm-Assumptions},
when we also let $\beta\equiv\beta(n)=\Omega(1)$, the \emph{time} to reach this ball scales as \[
\Omega\left(nr^{\gamma(d-2)}\frac{e^{R(n)d/\alpha}}{R(n)}\log n\right),\qquad\mathrm{a.a.s.}\]

where $R(n)\equiv\log(1+\beta(n))$.

\subsubsection{Impact of Increasing Communication Range on the Convergence Speed}

For a fixed $\beta$, from (\ref{eq:MixingSlotsLongRange}) we notice that as with short-range links, the slot mixing time scales polynomially in $r$ for $d>1$. The parameter $s$ that controls the distance of long-range communication enters the scaling law through $r^\gamma$, since $s=\Theta(r^\gamma)$. By comparing (\ref{eq:MixingSlotsLongRange}) and (\ref{eq:FinalMixingTimeSlots}) it is clear that its role is identical to that of $r$ in (\ref{eq:FinalMixingTimeSlots}). Thus we expect the impact of increased communication range to have the same dependence of the network dimension as in (\ref{eq:FinalMixingTimeSlots}). From Proposition \ref{prop:ScheduleLB} and an analysis similar to Section \ref{sub:ChoiceTransmRate}, it follows that while one-dimensional networks can converge faster from an increased communication range despite greater interference, the convergence speed of two-dimensional networks scales independently of the communication range. In higher-dimensions the increased interference from a larger communication range can actually lower the rate of convergence. 

From the model $s=\Theta(r^\gamma)$, which implies that a larger $s$ can result from either a larger $r$ (communicating with more nearby nodes) or a smaller $\gamma$ (communicating with nodes farther away). In either case we find that (\ref{eq:MixingSlotsLongRange}) scales faster than (\ref{eq:FinalMixingTimeSlots}): when interference is accounted for, selective long-range communications do not improve the rate of convergence.

\subsubsection{The Importance of Long-Range Clusters\label{sub:LengthAndDensityLRLinks}}
Here we discuss the importance of forming long links from a node to a \emph{cluster} of nodes. Briefly, we argue that adding only a few long edges to a given node does not take full advantage of the broadcast nature of the wireless medium: while these fewer long edges to a node reduce the spectral gap (and can increase mixing time as a result of Theorem \ref{thm:MixingTimeBounds}), forming long links from a node to a \emph{cluster} of nodes causes approximately the same interference as forming a point-to-point link of the same distance. Hence in lowering this cluster size, we do not gain from reduced interference, but can only worsen the spectral gap. So, when interference is factored in, allowing a node to talk to a far-off \emph{cluster} rather than a few far-off nodes allows faster mixing for the same level of interference. 
%

The effect of forming clusters is captured in the long-range model in Section \ref{sub:Tiered-Network-Model}, which adds all the nodes from a partner hypercube as long-range partners. This maximizes the number of long edges contributed by each hypercube and results in the scaling law in Proposition \ref{prop:ConductanceBound}. This is key to deriving Theorem \ref{prop:SpecGapMainResult}.

Suppose we modify the way long edges are added in this model by constructing a new graph $G_n^\prime(s)$ by assigning each node only $\rho_n = \mathcal{O}(nr^d)$ long-range partners in each partner hypercube. Evidently $G_n^\prime(s)$ is regular a.a.s., with node degree $nr^d|b(0,1)| + 2d\rho_n + o(1)$; so iterations as in (\ref{eq:ScalarStateUpdate}) converge to the average consensus point a.a.s. Denote the corresponding update matrix by $\mathbf{W}\!_n^{\,\prime}$. We will now examine the scaling of the spectral gap of $\mathbf{W}\!_n^{\,\prime}$.

Following the steps in the proof of Proposition \ref{prop:ConductanceBound}, the (edge-weighted) conductance of $G^\prime$ is $\Theta(r + r^{\gamma}(\rho_n/nr^d))$. Since $\rho_n = \mathcal{O}(nr^d)$, the conductance can scale no faster than $r^{\gamma}$. 

Therefore, unlike in the case with $G_n^{\mathrm{l}}(s)$, exploiting the inclusion $G_n^{\prime}(s)\subset G_n^{\mathrm{sh}}(s/2)$ is not enough to conclude the spectral gap of $\mathbf{W}\!_n^{\,\prime}$ to be $\Theta(r^{2\gamma})$. But the inclusion does confirm the spectral gap to be $\mathcal{O}(r^{2\gamma})$. Hence, as one would expect, iterations of the form (\ref{eq:ScalarStateUpdate}) can converge no faster with $\mathbf{W}\!_n^{\,\prime}$ than with $\mathbf{W}\!_n^{\,\mathrm{l}}$.

However, in the scaling limit, the interference resulting from the construction of $G_n^{\mathrm{l}}$ or $G_n^{\prime}$ are the same: it is obvious from Propositions \ref{prop:ScheduleLB} and \ref{prop:ScheduleUB} that the shortest feasible TDMA schedule for $G_n^{\prime}$ is also $\Theta(nr^{\gamma d})$ slots. We thus conclude that maximizing the cluster size to include all the nodes inside a partner hypercube speeds up convergence for the same level of interference. However, when this cluster is enlarged to include all nodes within a radius $s$, we have a disk graph with radius $s$. From the results in the previous sections, it is clear that the interference penalty to realize this larger disk graph scales similarly but is certainly larger than that of $G_n^{\mathrm{l}}(s)$, which has only selective long-range links.

\section{Conclusions\label{sec:Conclusions}}

We analyzed the convergence rate of average consensus algorithms in the scaling limit of dense wireless networks by combining results
from Markov chain theory, random geometric graphs, and wireless networks. When messages in a topology are exchanged over wireless links, the
impact of a greater communication range depends crucially on the network dimension. Increased communication range can speed up convergence in one-dimensional networks despite greater interference. In two-dimensional networks, the convergence speed scales independently of the communication range. In three- (and higher-) dimensional networks, forming long links can actually slow down convergence. These results hold whether each node only communicates over short links, or, additionally, with a cluster of far-away nodes. 

These results greatly differ from many optimistic results about the benefit of long-range connectivity obtained by analyzing the consensus problem in an abstract graph-theoretic setting. Our results underline the need to accurately account for the cost of interference in designing fast-converging topologies for the average consensus algorithm, or for distributed signal processing problems, in general.

\bibliographystyle{ieeetr} \bibliographystyle{ieeetr} \bibliographystyle{ieeetr}
\bibliographystyle{ieeetr}
\bibliography{./ConsensusRefs}

\end{document}